\documentclass[journal]{IEEEtran}
\IEEEoverridecommandlockouts
\usepackage{amsmath,amssymb,amsfonts,amsthm}
\usepackage{algorithmic}
\usepackage{graphicx}
\usepackage{relsize}
\usepackage{dsfont}
\usepackage{hyperref}

\ifCLASSOPTIONcompsoc \usepackage[caption=false,font=normalsize,labelfont=sf,textfont=sf]{subfig}
\else
\usepackage[caption=false,font=footnotesize]{subfig}
\fi

\usepackage{accents}
\usepackage{textcomp}
\def\BibTeX{{\rm B\kern-.05em{\sc i\kern-.025em b}\kern-.08em
    T\kern-.1667em\lower.7ex\hbox{E}\kern-.125emX}}

\usepackage{lipsum} 
 \usepackage{balance}
\usepackage{booktabs} 
\usepackage{enumitem}
\usepackage{mathtools}
\usepackage[lined,boxed,commentsnumbered,linesnumbered,ruled]{algorithm2e}
\usepackage{tikz} 
\usetikzlibrary{shapes, patterns,decorations.text, decorations.pathreplacing}
\usetikzlibrary{external}

\usepackage{soul}

\newcommand{\mycomment}[1]{}

\usepackage[noadjust]{cite} 

\usepackage{pgfplots}
\pgfplotsset{filter discard warning=false}
\pgfplotsset{compat=1.14}
\usepgfplotslibrary{colorbrewer}
\usepgfplotslibrary{groupplots}

\newtheorem{theorem}{Theorem}
\newtheorem{lemma}[theorem]{Lemma}

\newenvironment{subproof}[1][\proofname]{
  
  \begin{proof}[#1]
}{\end{proof}}

\usepackage[capitalize]{cleveref}
\crefname{equation}{\unskip}{\unskip}
\crefname{claim}{Claim}{Claims} 
\usepackage{array}
\usepackage{multirow}
\newcolumntype{C}[1]{>{\centering\arraybackslash}p{#1}}
\allowdisplaybreaks

\newcommand{\E}[1]{\mathbb{E}\left[ #1 \right]}

\newcommand{\CC}{\mathsf{CC}} 
\newcommand{\NBEC}{\mathsf{NBEC}} 
\newcommand{\nummotifs}{n}
\newcommand{\numchosen}{k}
\newcommand{\numreads}{R}
\newcommand{\numunique}{\ell}

\newcommand{\pinter}{\rho} 

\newcommand{\dv}{d_\mathsf{v}}
\newcommand{\dc}{d_\mathsf{c}}
\newcommand{\chainlength}{L_\mathsf{\rmine{p}}}
\newcommand{\blocklength}{N_\mathsf{\rfirst{p}}}
\newcommand{\rsc}{R_\mathsf{SC}}

\newcommand{\bx}{\mathbf{x}}
\newcommand{\by}{\mathbf{y}}
\newcommand{\bz}{\mathbf{z}}
\newcommand{\bs}{\mathbf{s}}
\newcommand{\cX}{\mathcal{X}}
\newcommand{\cY}{\mathcal{Y}}
\newcommand{\prob}[1]{\mathbb{P}\left(#1\right)}
\newcommand{\logtwo}[1]{\log_2\left[#1\right]}
\newcommand{\stirlingii}{\genfrac{\{}{\}}{0pt}{}}

\newcommand{\bPell}{\mathbf{P}^\ell}

\newcommand{\rcost}{p_\mathsf{r}}
\newcommand{\wcost}{p_\mathsf{w}}
\newcommand{\rwcostratio}{\lambda}
\newcommand{\numreadsopt}{\numreads^*}

\newcommand{\rfirst}[1]{#1}
\newcommand{\rsecond}[1]{#1}
\newcommand{\rthird}[1]{#1}
\newcommand{\rmine}[1]{#1}

\providecommand{\keywords}[1]
{
    {\small	\textbf{\textit{Index Terms---}#1.}}
}

\begin{document}

\title{Coding Over Coupon Collector Channels\\for Combinatorial Motif-Based DNA Storage}
\author{
    Roman Sokolovskii, Parv Agarwal, Luis Alberto Croquevielle, Zijian Zhou, and Thomas Heinis
    \thanks{This work was funded by MoSS (grant 101058035).}
    \thanks{The authors are with the Department of Computing, Imperial College London, SW7 2AZ London, UK (email:
    r.sokolovskii@imperial.ac.uk; p.agarwal23@imperial.ac.uk; a.croquevielle22@imperial.ac.uk; zijian.zhou18@imperial.ac.uk; t.heinis@imperial.ac.uk).}}

\maketitle

\begin{abstract} 
    Encoding information in combinations of pre-synthesised deoxyribonucleic acid (DNA) strands (referred to as \textit{motifs}) is an interesting approach to DNA storage that could potentially circumvent the prohibitive costs of nucleotide-by-nucleotide DNA synthesis. Based on our analysis of an empirical data set from HelixWorks, we propose two channel models for this setup (with and without interference) and analyse their fundamental limits. We propose a coding scheme that approaches those limits by leveraging all information available at the output of the channel, in contrast to earlier schemes developed for a similar setup by Preuss \textit{et al.} We highlight an important connection between channel capacity curves and the fundamental trade-off between synthesis (writing) and sequencing (reading) \rmine{costs}, and offer a way to mitigate an exponential growth in decoding complexity with the size of the motif library.
\end{abstract}

\section{Introduction}
\label{sec:intro}

The density and durability of deoxyribonucleic acid (DNA) make it an attractive medium for archival storage.
However, despite numerous demonstrations of the feasibility of DNA storage (we refer the reader to~\cite{ref:Ceze19} for a comprehensive overview), prohibitive cost of DNA synthesis hinders its widespread adoption.
A promising alternative to the most commonly used (and expensive) nucleotide-by-nucleotide DNA synthesis is encoding information in combinations of pre-synthesised oligonucleotides (referred to as \textit{motifs}) from a fixed motif library~\cite{ref:Yan23}. We refer to those combinations of motifs as \textit{combinatorial symbols}.
During synthesis (writing), the selected combinations of motifs are added to a pool of growing DNA strands, with no control over which of those motifs is attached to any given strand, and therefore over which of those motifs will be encountered during sequencing (reading).
The exponential growth in the number of combinations allows for a linear growth in information density with the size of the motif library, at the cost of additional sequencing.

Combinatorial motif-based DNA storage is a refreshing setup from a coding-theoretic perspective with little research conducted to date. In~\cite{ref:preuss2021efficient}, Preuss~\textit{et al.} emphasise an important connection between the process of accumulating the motifs during sequencing with the classical Coupon Collector's Problem. The authors employ maximum distance separable (MDS) codes and rely on fully accumulated symbols for decoding. In other words, the coding scheme requires only a subset of the combinatorial symbols to be recovered in order to decode a codeword, but these symbols must be recovered fully (i.e., all motifs constituting these symbols must be encountered during sequencing) for decoding to be possible. This MDS-based scheme is subsequently analysed in detail in a follow-up paper~\cite{ref:preuss2024sequencing} that proposes a computational framework for choosing the parameters of the coding scheme based on coverage and reliability constraints. \rmine{For a similar setup,~\cite{ref:coh2024opti}  gives the expected number of reads required to reconstruct information.} \rthird{We remark that in the case of standard (as opposed to combinatorial motif-based) synthesis, the connection between DNA sequencing and the Coupon Collector's Problem is formally introduced in~\cite{ref:bar2023cover}; a random-access version of the problem is further investigated in~\cite{ref:gruica2024reducing,ref:abraham2024covering}.}


The approach in~\cite{ref:preuss2021efficient,ref:preuss2024sequencing} has two important limitations: First, as we elaborate in Section~\ref{sec:channel}, using only fully recovered \rmine{symbols} during decoding reduces the pipeline to an erasure channel; processing partial combinations offers a promising way to improve decoding by leveraging all available information. Second, the setup in~\cite{ref:preuss2024sequencing} mitigates possible \textit{interference} (i.e., the detection of a motif that is not a part of the chosen combination) by requiring several copies of the candidate motifs to be encountered and declaring erasure otherwise. This further increases the required coverage---again, processing full available probabilistic information could significantly improve decoding performance.

\rmine{The first work to make use of partial combinations is~\cite{ref:saba2024error}, which studies adversarial code constructions that provably recover from up to $e$ missing motifs in up to $d$ symbols (the possibility of interference is not considered). The work provides bounds on the size of such codes and offers an explicit encoding scheme for correcting up to $e=1$ missing motif (with $d$ remaining a configurable parameter) when the block length does not exceed the size of the motif library.}

\begin{figure*}[t!]
    \centering
    \includegraphics[width=\textwidth,height=0.75\textheight,keepaspectratio]{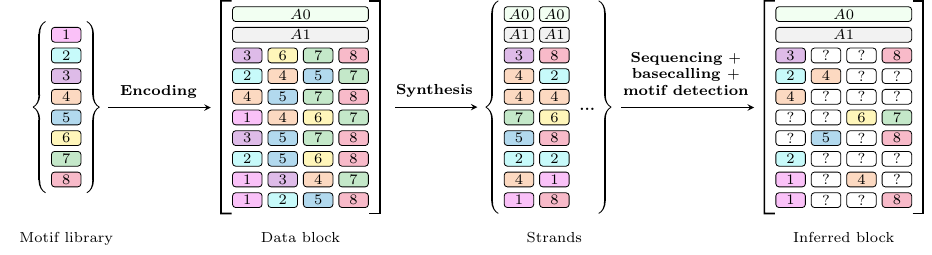}
    \vspace{-15pt}
    \caption{
    Combinatorial DNA data storage pipeline used in~\cite{ref:Yan23}. First, the data is encoded into blocks, each containing $2$ address-carrying symbols (denoted by $A0$ as $A1$) and $8$ payload-carrying symbols; a \rmine{payload} symbol (shown as a row of the data block) is a combination of $4$ distinct motifs chosen out of a library of $8$. 
    Second, the DNA pool is synthesised: in each payload synthesis cycle, the $4$ chosen motifs corresponding to the synthesised symbol are added to the reaction tube containing a set of growing strands (shown vertically as columns); a random motif is attached to each strand.
    The bioinformatics pipeline at the reading stage includes sequencing, basecalling, and motif detection. It produces the original data blocks with some motifs missing or interfering due to synthesis and detection errors. }
    \label{fig:protocol}
    \vspace{-2ex}
\end{figure*}

In this work, we consider error correction for combinatorial motif-based DNA storage beyond the MDS-based schemes proposed and analysed in~\cite{ref:preuss2021efficient,ref:preuss2024sequencing}, \rmine{aiming to use \textit{all} information available at the output of the channel}. Based on our analysis of the experimental data set provided by HelixWorks (Section~\ref{sec:motif_storage}), we formally introduce a channel model that mimics the empirically observed behaviour, which we call the Coupon Collector Channel, either with or without interference (Section~\ref{sec:channel}). We derive the capacity for the interference-free Coupon Collector Channel (Section~\ref{sse:capacity}) and its erasure-channel counterpart (Section~\ref{sse:erasure}), as well as propose a computational method to estimate the capacity of the channel with interference (Section~\ref{sse:interference_capacity}). We demonstrate that relying only on fully accumulated combinations imposes a fundamental limit on achievable performance. In Section~\ref{sse:read_write_cost}, we highlight an important connection between the capacity curves and the fundamental trade-off between the costs of reading and writing \rmine{based on the total-cost optimisation framework introduced in~\cite{ref:heck2017fund}}. Section~\ref{sec:ldpc} offers a way to approach the fundamental limits of the Coupon Collector Channel using protograph-based non-binary spatially coupled low-density parity-check (SC-LDPC) codes, which we demonstrate numerically in Section~\ref{sec:results}, and describes how to avoid an exponential growth in decoding complexity with the size of the motif library.




\section{Motif-Based DNA Storage}
\label{sec:motif_storage}

\subsection{General Principle}
\label{sse:general_principle}

The starting point of our work is the combinatorial motif-based DNA storage pipeline used by HelixWorks and described in detail in~\cite{ref:Yan23}. We illustrate it schematically in Fig.~\ref{fig:protocol}. Instead of individual nucleotides, synthesis operates with a library of pre-synthesised $50$-nucleotide motifs. The size of the library determines \rfirst{the} information density; a library of $8$ motifs is used in current real-world experiments by HelixWorks, and we use their parameters as our running example throughout this paper without loss of generality. The data is arranged into blocks, where each block contains $2$ symbols for addressing (denoted by $A0$ and $A1$ in Fig.~\ref{fig:protocol}) and $8$ combinatorial symbols for the payload. Each payload symbol, in turn, is a combination of $4$ distinct motifs out of $8$ in the library.\footnote{\rmine{A symbol used for addressing comprises a single motif from the library; we leave the problem of addressing outside the scope of this paper.}}  For example, the first payload symbol in the data block illustrated in Fig.~\ref{fig:protocol} is $\{3, 6, 7, 8\}$ (the symbols themselves are arranged from top to bottom). During synthesis, a block is represented by a set of growing strands. A synthesis cycle incorporates one combinatorial symbol; specifically, the $4$ motifs encoding the symbol are added to the pool and appended to the strands. Carrying on with the example in Fig.~\ref{fig:protocol}, we show two synthesised strands (a strand is shown as a column): during the first payload synthesis cycle, multiple copies of each of the motifs in the symbol $\{3,6,7,8\}$ are added to the pool; during the second cycle, the motifs $\{2,4,5,7\}$ are added, and so on. By design, only one out of the four added motifs gets attached to a strand in one cycle; however, there is no control over which one it will be. This fundamental randomness is the key property of motif-based DNA storage. In our running example in Fig.~\ref{fig:protocol}, motif $3$ happens to get attached to the first strand in the first payload synthesis cycle, and motif $8$ is attached to the second strand.

At the reading stage, the DNA strands are sequenced, and a bioinformatics pipeline that includes basecalling and motif detection attempts to reconstruct the stored data blocks. However, limited sequencing coverage and errors during basecalling and motif detection may lead to incomplete data retrieval and erroneous motif detection. Assuming that only the two strands displayed in Fig.~\ref{fig:protocol} are sequenced and that all motifs are correctly detected, only motifs $3$ and $8$ from the first payload cycle are observed at the reading end of the pipeline---the other two motifs remain unknown and could in principle be any two distinct motifs from the set $\{1,2,4,5,6,7\}$.


\vspace{-1ex}
\subsection{Experimental Data Set}
\label{sse:dataset}



\begin{figure*}[t]
    \centering
    \includegraphics[width=\textwidth]{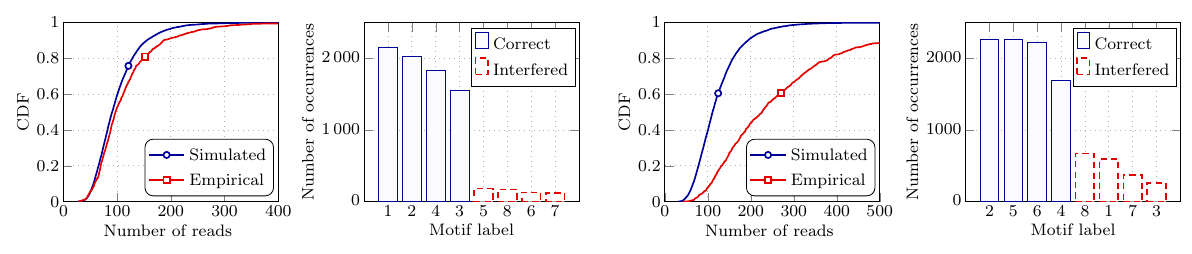}
    
    \vspace{-5pt}
    \caption{
    Empirical and simulated uncoded data recovery CDFs and motif histograms. The first two subplots correspond to block $21$. The first subplot shows a close match between the simulated (blue curve, circle marker) and empirical (red curve, square marker) CDF. The representative interfered motif histogram for symbol $8$ in this block (second subplot, red dashed bars) is relatively even. The last two plots correspond to block $1$. A large discrepancy between the simulated and empirical CDF (third plot) is caused by imbalanced interfered motif histogram (last plot, representative symbol $8$, red dashed bars).
    }
    \label{fig:grouped}
    \vspace{-2ex}
\end{figure*}

To investigate the empirical behaviour of the synthesis-storage-sequencing pipeline, we analysed a data set generated by a wet-lab experiment conducted by HelixWorks.
The experiment synthesised $21$ blocks. The resulting DNA pool was sequenced via \textit{Oxford Nanopore Technologies} (ONT) nanopore sequencing. The generated reads were basecalled and processed by HelixWorks's software pipeline, which detected $3.7$ million motif occurrences. The pipeline assigns reads to blocks using ONT's barcoding feature based on the two address symbols, and detects payload motifs using exact search for the $50$-nucleotide motif sequences in the library.\footnote{We are aware that more elaborate motif detection methods are being developed by HelixWorks and collaborators.}   Our analysis shows three distinct patterns in the input-output error behaviour, which we discuss in more detail in Section~\ref{sse:simulator}:

\begin{enumerate}
    \item As was highlighted in Section~\ref{sse:general_principle}, there is a fundamental uncertainty about which motif is attached during synthesis and read during sequencing. This means that accumulating the motifs at the reading stage resembles the classical Coupon Collector's Problem, as discussed in~\cite{ref:preuss2021efficient}.
    
    \item In addition to the coupon-collecting dynamics of the reading process, we have also observed significant variation in how often the motifs belonging to different locations within a block are detected, which we refer to as cycle visibility. In particular, the first motif in a block is detected $5.7$ times more often than the last, with a general monotonically decreasing trend from the beginning to the end of the block. This visibility variation suggests that there is an accumulation of synthesis failures or of detection errors in the pipeline.
    
    \item Finally, we observe some interference in the form of motifs which are erroneously detected in symbols that do not contain them. Different motifs are observed at different rates, and this pattern does not appear to be consistent across different blocks. We speculate that these errors are primarily caused by occasional incorrect assignment of entire reads to wrong blocks by ONT's barcoding.
\end{enumerate}

\subsection{Channel Simulator}
\label{sse:simulator}



To test whether the error patterns described in Section~\ref{sse:dataset} are sufficient to model the dynamics of the sequencing process, we developed a simulator that replicates those patterns using parameters derived from empirical data. For a given block and a given number of reads, we compare the probability of recovering the entire block (with no error-correction coding) for the simulated reads and the reads obtained by sub-sampling the empirical data set down to the required coverage.

For each read, the simulator first uses uniform sampling to generate a sequence of $8$ motifs, each motif chosen from the $4$ motifs of the corresponding symbol. Then, it samples from the empirical cycle visibility distribution to determine which (if any) of these motifs are visible at the reading stage. Finally, based on the empirical average interference rate, we emulate interference by uniform sampling from \textit{all} motifs in the library (in contrast to sampling only from the motifs belonging to the chosen symbol).

Fig.~\ref{fig:grouped} focuses on two representative blocks from the empirical data set: the first two plots correspond to block $21$, and the last two plots correspond to block $1$. For each block, we plot the cumulative distribution function (CDF) of the probability of fully recovering the entire block and a representative empirical histogram of the number of detected motifs for one of the symbols within the block. In the first (leftmost) subplot (block $21$), we observe a close match between the simulated (blue curve, circle marker) and empirical (red curve, square marker) CDFs. We picked symbol $8$ from that block and plotted the number of times each motif in the library was detected in that location (second subplot), splitting the motifs into those that belong to the chosen symbol (blue bars) and those that do not (red dashed bars). The distribution of interfered motifs (red dashed bars) is relatively even. Conversely, for block $1$, where there is a significant mismatch between the simulated and empirical CDFs (third plot), the typical distribution of interfered motifs (symbol $8$, last plot, red dashed bars) is imbalanced. We observe this \rmine{behaviour} more generally, with around half of the blocks in the data set showing a close match to the simulated model and the other half showing a significant mismatch with the associated imbalance in interfered motif distributions. \rfirst{We have not discerned any systematic pattern in which blocks exhibit this mismatch}.\footnote{We remark that when we artificially filter out interfered motifs, we observe a near-perfect match between the simulated and empirical CDFs for all blocks.}

We conclude that when the pattern of interference is relatively uniform, the channel simulator provides a good model of the synthesis-storage-sequencing pipeline. In this work, we focus on the coupon-collecting nature of the channel and assume uniform interference, leaving unequal cycle visibility and skewed interference outside the scope of the paper. Tackling these aspects of the pipeline is an interesting direction for future work.
\rsecond{We remark that the channel models proposed in Section~\ref{sec:channel} aim to reflect the empirically observed behaviour (see Section~\ref{sse:dataset}) while also being applicable to other motif-based DNA storage schemes beyond the particular pipeline employed by HelixWorks.}
\section{Coupon Collector Channels}
\label{sec:channel}


\begin{figure*}[ht!]
    \centering
    \includegraphics[width=\textwidth]{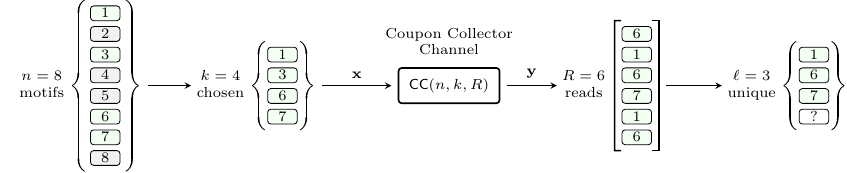}
    \vspace{-15pt}
    \caption{An example of transmission over the Coupon Collector Channel $\CC(\nummotifs=8,\numchosen=4,\numreads=6)$. From a library of $\nummotifs=8$ motifs, a combination of $\numchosen=4$ is chosen for transmission. The Coupon Collector Channel generates $\numreads=6$ reads by uniform sampling with replacement among the $\numchosen$ chosen motifs. In the illustrated example, only $\numunique=3$ motifs out of $\numchosen=4$ are encountered.}
    \label{fig:channel_model}
    \vspace{-2ex}
\end{figure*}

The simulator described in Section~\ref{sse:simulator} serves as the starting point for channel modelling. We are now ready to give a formal definition of the channel model without interference (we discuss interference in Section~\ref{sse:interference_capacity}).
Let there be a set $\mathcal{M}$ containing $n$ different motifs and let $k$ be a positive integer such that $k < n$. Then any subset $\bx \subset \mathcal{M}$ with exactly $k$ motifs is a combinatorial symbol. Let $\cX$ be the set of possible combinatorial symbols. It is easy to see that $|\cX| = \binom{n}{k}$.

Transmission through the channel involves a parameter $R$, which represents the number of reads and must be a positive integer. In its simplest version, the channel is denoted as $\CC(n,k,R)$ and represents interference-free sequencing. In this case, when a symbol $\bx$ is transmitted, the received sequence $\by$ is generated by sampling $R$ motifs from $\bx$, independently and uniformly at random. Figure \ref{fig:channel_model} illustrates this channel.

Notice that the received sequence can be written as $\by = (y_1, \ldots, y_R)$ where $y_i\in\bx$ for all $i=1,\ldots,R$. This means that the number of unique \rfirst{(i.e., distinct)} motifs present in $\by$ can range from $1$ to \rfirst{$L=\min\{k,R\}$}. We denote by $\cY$ the set of all possible received sequences---i.e., the set of $R$-tuples with at least $1$ and at most $L$ different motifs. \rfirst{Further, we say that $\bx$ and $\by$ are \textit{compatible} if $y_i \in \bx$ for all $i = 1, \ldots, R$, and we use $\cY(\bx)$ to denote the set of all $\by$ that are compatible with $\bx$.}

\subsection{Capacity of the $\CC(\nummotifs,\numchosen,\numreads)$ Channel}
\label{sse:capacity}


We first give the capacity of the channel as discussed so far.

\begin{theorem}
    \label{thm:capacity-cc-channel-no-interference}
    The capacity of $\CC(n,k,R)$ is given by 
    \begin{equation}
        \label{eq:capacity_cc}
        \mathcal{C}_\CC
        =
        \log_2\binom{n}{k}
        -
        k^{-R}
        \sum_{\ell=1}^L
        \binom{k}{\ell}\stirlingii{R}{\ell}\,\ell!\,
        \log_2\binom{n-\ell}{k-\ell}\,,
    \end{equation}
    where $L=\min\{k, R\}$, and \rfirst{$\stirlingii{R}{\ell} = 1/\ell!\sum_{i=0}^\ell (-1)^{\ell-i} \binom{\ell}{i} i^R$} are \rmine{the} Stirling numbers of the second kind.
\end{theorem}

Before proving this result, we introduce important notation and a useful Lemma. Notice that when a symbol $\bx$ is transmitted, only \rfirst{a compatible} $\by\in\cY(\bx)$ can be received, each being equally likely. Now, the set $\cY(\bx)$ can be partitioned into subsets that are easier to work with, defined as
\begin{equation}
    \label{eq:defn_y_ell}
    \cY_\ell(\bx) = \{
        \by\in\cY(\bx)
        \>|\>
        \by \text{ contains } \ell \text{ unique motifs}
    \}\,.
\end{equation}
Notice that $\{\cY_\ell(\bx)\}_{\ell=1}^L$ is a partition of $\cY(\bx)$. The following result gives the size of these sets.

\begin{lemma}
\label{lemma:cc-channel-partition-size}
    Let $\bx \in \cX$ and $\ell \in \{1,\ldots,L\}$ with $L=\min\{k,R\}$. Then it holds that
    \begin{equation}
        \label{eq:size_y_ell}
        |\cY_\ell(\bx)| = \binom{k}{\ell}\stirlingii{R}{\ell}\,\ell! \,,
    \end{equation}
    where $\stirlingii{R}{\ell}$ are Stirling numbers of the second kind.
\end{lemma}
\begin{proof}
    On the one hand, there are $\binom{k}{\ell}$ ways to choose $\ell$ unique motifs out of $k$ motifs in $\bx$. On the other, there are $\stirlingii{R}{\ell}\ell!$ ways to distribute $R$ distinct reads among $\ell$ distinct motifs (see case 3 of The Twelvefold Way in \cite[Section 1.9]{ref:stanley2011enumerative}). Multiplying both factors gives the result.
\end{proof}

Consider two random variables $X, Y$ such that $X$ represents the transmitted combinatorial symbol and $Y$ represents the corresponding received sequence.
We use the following compact notation:
$\prob{\bx}$ denotes $\prob{X=\bx}$, and analogously for $\prob{\by}$. $\prob{\bx | \by}$ denotes $\prob{X=\bx | Y=\by}$, and analogously for $\prob{\by | \bx}$. Finally, $\prob{\bx, \by}$ denotes $\prob{X=\bx, Y=\by}$.

\begin{proof}[Proof of Theorem \ref{thm:capacity-cc-channel-no-interference}]
     
     We first prove that the mutual information $I(X;Y)$ under the uniform input distribution is equal to the right-hand side of~\eqref{eq:capacity_cc}. Second, we prove that the uniform distribution is capacity-achieving.

    \begin{subproof}[Uniform input distribution]
         Assume $X$ to be uniformly distributed. We expand $I(X;Y)=H(X)-H(X|Y)$ and calculate the two entropy terms. First, since $|\cX|=\binom{n}{k}$ and $\prob{\bx} = 1/|\cX|$ for all $\bx$, we get
        \begin{equation}
            \label{eq:cc-channel-entropy-X}
            H(X)
            =
            -\sum_{\bx\in\cX} \prob{\bx}\log_2\prob{\bx}
            =
            \log_2\binom{n}{k}\,.
        \end{equation}
        Second, we need to compute $H(X|Y)$.
        By definition of the channel, it is easy to see that $\prob{\by|\bx}=k^{-R}$ if $\by\in\cY(\bx)$, and $\prob{\by|\bx}=0$ otherwise. Using this, we get
        \begin{align}
            H(X|Y)
            &=
            -
            \sum_{\bx\in\cX}\sum_{\by\in\cY}
            \prob{\bx,\by}
            \log_2\prob{\bx|\by} \nonumber \\
            &=
            -
            \sum_{\bx\in\cX}\prob{\bx}
            \sum_{\by\in\cY(\bx)}
            \prob{\by|\bx}\log_2\prob{\bx|\by} \nonumber \\
            &=
            -
            \sum_{\bx\in\cX}\prob{\bx}
            \sum_{\by\in\cY(\bx)}k^{-R}\log_2\prob{\bx|\by}\,.
            \label{eq:cc-channel-entropy-X-given-Y-v1}
        \end{align}
        Notice that $\prob{\bx|\by}$ only depends on the number of unique motifs within $\by$. Also, it is clear that if $\by$ has $\ell$ unique motifs, then there are $\binom{n-\ell}{k-\ell}$ possible symbols in $\cX$ that are compatible. That is, $\prob{\bx|\by} = 1/\binom{n-\ell}{k-\ell}$. Replacing in (\ref{eq:cc-channel-entropy-X-given-Y-v1}):
        \begin{align}
            H(X|Y)
            &=
            k^{-R}
            \sum_{\bx\in\cX}\prob{\bx}
            \sum_{\ell=1}^L\sum_{\by\in\cY_\ell(\bx)}
            \log_2\binom{n-\ell}{k-\ell} \nonumber \\
            &=
            \label{eq:cc-channel-entropy-X-given-Y-v2}
            k^{-R}
            \sum_{\bx\in\cX}\prob{\bx}
            \sum_{\ell=1}^L |\cY_\ell(\bx)|\,
            \log_2\binom{n-\ell}{k-\ell}\,.
        \end{align}

        We can substitute the value given in~\eqref{eq:size_y_ell} for $|\cY_\ell(\bx)|$ into~\eqref{eq:cc-channel-entropy-X-given-Y-v2}. Since the inner sum does not depend on the value of $\bx$, we get\footnote{We remark that~\eqref{eq:cc-channel-entropy-X-given-Y-v3} can be interpreted as averaging the conditional entropy of $X$ over the distribution of the number of unique motifs $\ell$ observed in $Y$.}
        \begin{equation}
            \label{eq:cc-channel-entropy-X-given-Y-v3}
            H(X|Y)
            =
            k^{-R}
            \sum_{\ell=1}^L \binom{k}{\ell}\stirlingii{R}{\ell}\,\ell!\,
            \log_2\binom{n-\ell}{k-\ell}\,.
        \end{equation}
        Combining~\eqref{eq:cc-channel-entropy-X} and \eqref{eq:cc-channel-entropy-X-given-Y-v3}, we conclude that under uniform input distribution $I(X;Y)$ is equal to the right-hand side of~\eqref{eq:capacity_cc}.
    \end{subproof}

    \begin{subproof}[Symmetry]
        It remains to prove that the uniform input distribution maximises $I(X;Y)$. For this, we rely on the following result from \cite{ref:gallager1968information}:

        \begin{lemma}[{\cite[Theorem 4.5.2]{ref:gallager1968information}}]
        \label{lemma:capacity-symmetric-channels}
            For a symmetric discrete memoryless channel, capacity is achieved by using the inputs with equal probability.
        \end{lemma}
        
        We can therefore prove that the uniform input distribution achieves capacity by proving that $\CC(n,k,R)$ is symmetric---i.e., that the set of channel outputs can be partitioned in such a way that for each subset the corresponding matrix of transition probabilities $\prob{\by | \bx}$ has the property that each row is a permutation of each other row and each column is a permutation of each other column~\cite{ref:gallager1968information}. 
        
        For this, consider the partition $\{\cY_\ell\}_{\ell=1}^L$ of $\cY$, where $\cY_\ell$ is the set of all $\by\in\cY$ such that $\by$ has $\ell$ unique motifs. Now, let $\ell\in \{1,\ldots,L\}$ and consider the transition probabilities $\prob{\by | \bx}$ where  $\by\in\cY_\ell$ and $\bx\in\cX$. These probabilities can be written as a matrix $\bPell$ where element $\bPell_{ij}$ is the probability of receiving the $j$-th element of $\cY_\ell$ when the $i$-th element of $\cX$ is transmitted (consider any valid total order over $\cY_\ell$ and $\cX$). The elements $\bPell_{ij}$ can take only one of two possible values: $\bPell_{ij}=k^{-R}$ if $\bx$ and $\by$ are compatible, and $\bPell_{ij}=0$ otherwise.
    
        For any given $\by\in\cY_\ell$, there are $\binom{n-\ell}{k-\ell}$ compatible $\bx$ and $\binom{n}{k} - \binom{n-\ell}{k-\ell}$ non-compatible $\bx$. On the other hand, for any given $\bx\in\cX$, there are $|\cY_\ell(\bx)|$ compatible $\by$ and $|\cY_\ell| - |\cY_\ell(\bx)|$ non-compatible $\by$. In other words, and using Lemma \ref{lemma:cc-channel-partition-size} to get the value of $|\cY_\ell(\bx)|$:
        \begin{itemize}
            \item Every column in $\bPell$ has $\binom{n}{k}$ entries. Of these, $\binom{n-\ell}{k-\ell}$ have value $k^{-R}$, and the rest are $0$. This is the same for all columns, so it holds that all columns in $\bPell$ are permutations of each other.
            \item Every row in $\bPell$ has $|\cY_\ell|$ entries. Of these, $\binom{k}{\ell}\stirlingii{R}{\ell}\ell!$ have value $k^{-R}$, and the rest are $0$. This is the same for all rows, so it holds that all rows in $\bPell$ are permutations of each other.
        \end{itemize}
    
        These last two conditions constitute the definition of a symmetric channel, since $\{\cY_\ell\}_{\ell=1}^L$ is a partition of $\cY$. We conclude that the $\CC(n,k,R)$ channel is symmetric.
    \end{subproof}
    Together, the expression for $I(X;Y)$ under uniform input distribution and the symmetry of the channel prove Theorem~\ref{thm:capacity-cc-channel-no-interference}.
\end{proof}


\vspace{-4ex}
\subsection{Erasure Version of the Coupon Collector Channel}
\label{sse:erasure}

Suppose now that the decoder \rfirst{relies only on outputs $\by$ that fully determine the corresponding input $\bx$---in other words, on $\by\in\cY_\numchosen$.}
In this case, the channel effectively reduces to the non-binary erasure channel with erasure probability
\begin{equation}
    \label{eq:erasure_probability}
    1 - F(\numchosen,\numreads) = 1 - \stirlingii{\numreads}{\numchosen}\,\numchosen!\,\numchosen^{-\numreads}\,,
\end{equation}
where $F(\numchosen,\numreads)$ is the CDF in the standard Coupon Collector's Problem, denoting the probability that $\numreads$ trials are sufficient to collect all $\numchosen$ coupons~\cite[Example II.11]{ref:flajolet2009analytic}.
We denote the resulting all-or-nothing version of the Coupon Collector Channel as $\NBEC(\nummotifs,\numchosen,\numreads)$.

The capacity of $\NBEC(\nummotifs,\numchosen,\numreads)$ is a scaled CDF from the Coupon Collector's Problem $F(\numchosen,\numreads)$:%
\begin{equation}
    \label{eq:capacity_nbec}
    \mathcal{C}_\NBEC
    =
    \log_2\binom{n}{k} \cdot
   \stirlingii{\numreads}{\numchosen}\,\numchosen!\,\numchosen^{-\numreads}\,.
\end{equation}

\begin{figure}[!t]
    \centering
    \includegraphics[width=\columnwidth]{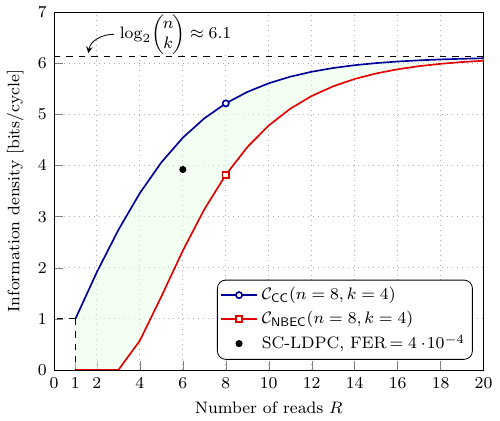}
    \vspace{-18pt}
    \caption{Capacities of $\CC$ (blue, circle marker) from~\eqref{eq:capacity_cc} and of $\NBEC$ (red, square marker) from~\eqref{eq:capacity_nbec} for $(\nummotifs=8,\numchosen=4)$. The horizontal dashed line shows the asymptotic limit for both curves. The region of operation enabled by full processing of partial information is shaded green. The black dot shows the operating point of the coding scheme proposed in Section~\ref{sec:ldpc}, see Section~\ref{sec:results} for details. The code operates well above the capacity of $\NBEC$.}
    \label{fig:capacity}
    \vspace{-2ex}
\end{figure}

Ignoring partial information available at the output of the Coupon Collector Channel entails a loss in achievable information rate for a given $\numreads$, as we illustrate in Fig.~\ref{fig:capacity}. The blue curve (circle marker) shows the capacity of $\CC(\nummotifs=8,\numchosen=4,\numreads)$ calculated using~\eqref{eq:capacity_cc}, and the red curve (square marker) shows the corresponding $\mathcal{C}_\NBEC$ from~\eqref{eq:capacity_nbec}. Both curves tend to $\log_2\binom{\nummotifs}{\numchosen}$ as $\numreads\to\infty$ (horizontal dashed line), where redundancy is no longer required.
A key contribution of this paper (see Section~\ref{sec:ldpc} below) is a coding scheme that uses all information available in $\by$ and is thus able to operate above $\mathcal{C}_\NBEC$ (i.e., in the area shaded green in Fig.~\ref{fig:capacity}), beyond the reach of the schemes that wait for all $\numchosen$ motifs to be observed, including the MDS-based schemes proposed and analysed in~\cite{ref:preuss2021efficient,ref:preuss2024sequencing}. The black dot in Fig.~\ref{fig:capacity} shows the operating point of a particular coding scheme described in detail in Section~\ref{sec:results}, demonstrating that it is indeed possible to operate in the green-shaded region in practice.




\vspace{-1ex}
\subsection{Coupon Collector Channel with Interference}
\label{sse:interference_capacity}

As we have observed in the experimental data set (Section~\ref{sse:dataset}), the bioinformatics pipeline may sometimes erroneously detect a motif that was not present in the transmitted symbol $\bx$.
Such errors are most likely caused by an incorrect assignment of a read to a wrong block and can therefore be though of as cross-block interference.

We incorporate this \rfirst{type of error} into the Coupon Collector Channel by introducing an auxiliary parameter $\pinter$ that controls the rate of interference. Specifically, given $\bx \in \cX$, each element $y_i$ in $\by = (y_1, \ldots, y_R)$ is \rfirst{independently} generated as follows: with probability $1 - \pinter$, $y_i$ is sampled from $\bx$ uniformly at random as for the standard $\CC(\nummotifs,\numchosen,\numreads)$ introduced earlier; otherwise, with probability $\pinter$, $y_i$ is sampled uniformly at random from the entire motif library $\mathcal{M}$ without regard to $\bx$, modelling interference from a different block. \rfirst{Importantly, we do \textit{not} sample from $\mathcal{M} \setminus \bx$ but from the entire $\mathcal{M}$ when we emulate interference because an interfered read from a wrong block may happen to contain a motif that is also present in $\bx$ by coincidence.} We denote the Coupon Collector Channel with interference by $\CC(\nummotifs,\numchosen,\numreads,\pinter)$.

We \rmine{can} estimate the capacity of $\CC(\nummotifs,\numchosen,\numreads,\pinter)$ \rfirst{using uniform input distribution (our model of interference does not break the symmetry of the channel). We rewrite} $I(X;Y)$ as
\begin{equation}
    \label{eq:mutual_information_expectation}
    I(X;Y) = H(X) - \mathbb{E} \Big[ H(X | Y = \by) \Big]\,,
\end{equation}
\rfirst{calculate $H(X)$ as in~\eqref{eq:cc-channel-entropy-X}, and estimate} the expectation using samples $\by$ generated via Monte-Carlo \rfirst{simulations} of the channel. \rfirst{Specifically, we generate a number of samples $\by$, calculate $H(X | Y = \by)$ for each sample as described below, average the obtained values to estimate $\mathbb{E} \big[ H(X | Y = \by)\big]$, and use it in~\eqref{eq:mutual_information_expectation} to obtain an estimation of the capacity of $\CC(\nummotifs,\numchosen,\numreads,\pinter)$.}

It remains to show how to calculate
\begin{equation}
    \label{eq:cond_entropy_given_x}
    H(X | Y = \by) = - \sum_{\bx \in \cX} \prob{\bx | \by}\log_2 \prob{\bx | \by}\,,
\end{equation}
which amounts to calculating $\prob{\bx|\by}$ for a given $\by$ and all $\bx \in \cX$. We observe that
\begin{equation}
    \label{eq:cond_prob_x_y}
    \prob{\bx | \by} = \prob{\by | \bx} \frac{\prob{\bx}}{\prob{\by}} \propto \prob{\by | \bx}\,,
\end{equation}
since both $\prob{\bx}$ and $\prob{\by}$ are constant for uniform input distribution and a fixed $\by$. This means that $\prob{\bx | \by}$ can be obtained by normalising $\prob{\by | \bx}$ by $\sum_{\bx' \in \cX} \prob{\by | \bx'}$, and we are left with the task of calculating $\prob{\by | \bx}$. We have
\begin{equation}
    \label{eq:cond_prob_y_x_prod}
    \prob{\by | \bx} = \prod_{i=1}^{\numreads} \prob{y_i | \bx}
\end{equation}
with
    \begin{align}
        \prob{y_i | \bx} &= \pinter \cdot \prob{y_i | \bx, \text{int.}} + (1 - \pinter) \cdot \prob{y_i | \bx, \text{no int.}}  \label{eq:cond_prob_elem_y_x} \\
        &= \left\{
        \begin{aligned}
            \pi_\mathsf{in} &= \pinter \cdot \frac{1}{\nummotifs} + (1 - \pinter) \cdot \frac{1}{\numchosen}\,,  & &\text{if~} y_i \in \bx \,;\\
            \pi_\mathsf{out} &= \pinter \cdot \frac{1}{\nummotifs}\,, & &\text{otherwise.} \\
        \end{aligned}
        \right. \nonumber
    \end{align}

Putting it all together, for a given $\by$ and each $\bx$, we need to count the number \rfirst{$s = |\left\{i: y_i \in \bx\right\}|$} of motifs in $\by$ that match the motifs in $\bx$, calculate
\begin{equation}
    \label{eq:cond_prob_y_x}
    \prob{\by | \bx} = \pi_\mathsf{in}^s \, \pi_\mathsf{out}^{\numreads - s}\,,
\end{equation}
normalise the resulting vector of probabilities \rfirst{$\left[\prob{\by | \bx}\right]_{\bx \in \cX}$} to sum up to~$1$, and use \rmine{$\left[\prob{\bx | \by}\right]_{\bx \in \cX}$ thus obtained} \rsecond{when calculating}~\eqref{eq:cond_entropy_given_x} to estimate~\eqref{eq:mutual_information_expectation} using simulated samples $\by$.
We remark that this procedure does not scale well with increasing $\nummotifs$ and $\numchosen$ because it requires considering all $|\cX| = \binom{n}{k}$ possible $\bx$ to calculate~\eqref{eq:cond_entropy_given_x}; providing a more efficient way to estimate (or compute exactly) the capacity of $\CC(\nummotifs,\numchosen,\numreads,\pinter)$ is therefore an interesting research problem that falls outside the scope of this work.

\begin{figure}[!t]
    \centering
    \includegraphics[width=\columnwidth]{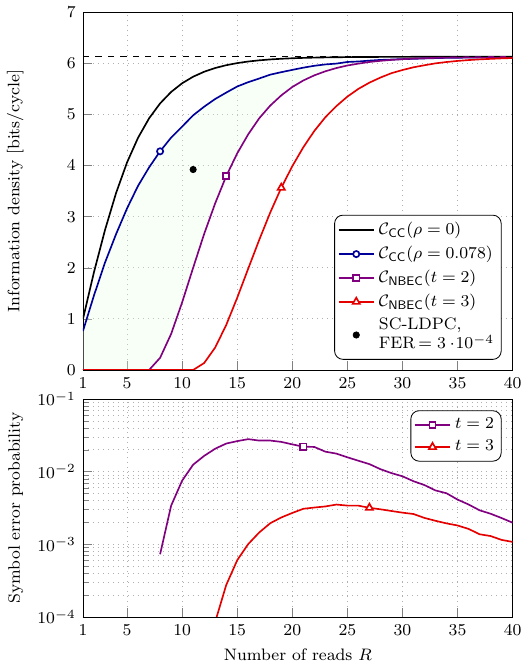}
    \vspace{-18pt}
    \caption{Top: capacity curves for the Coupon Collector Channel with $(\nummotifs=8,\numchosen=4)$, no interference (black curve, no marker) and $\pinter=0.078$ (blue curve, circle marker). For the hard-decision decoding scheme that waits for $t$ copies of the top $\numchosen$ motifs before making a decision, the capacity of the associated non-binary erasure channel (neglecting interference) is shown for $t=2$ (purple curve, square marker) and $t=3$ (red curve, triangle marker). Hard-decision schemes that mitigate interference by accumulating $t > 1$ copies of top-$\numchosen$ motifs cannot operate in the green-shaded region even when no interference is present. In contrast, the black dot shows the operating point of the coding scheme proposed in Section~\ref{sec:ldpc} for $\pinter=0.078$, see Section~\ref{sec:results} for details. Soft information processing allows the scheme to operate in the green-shaded region. Bottom: symbol-level substitution rates for the same hard-decision schemes with $t = 2$ and $t = 3$ as in the top figure (same line styles) and $\pinter=0.078$. These errors are not accounted for in the corresponding capacity curves, making the green-shaded region underestimate the potential gain from soft-decision information processing.}\label{fig:capacity_interference}
    \vspace{-2ex}
\end{figure}

In Fig.~\ref{fig:capacity_interference} (top), we plot the resulting capacity estimation for $\CC(\nummotifs=8,\numchosen=4,\numreads,\pinter)$ with $\pinter = 0.078$, the level of interference observed in the empirical data set (see Section~\ref{sse:dataset}), as the blue curve with the circle marker. Naturally, interference induces a penalty with respect to the corresponding interference-free channel (black curve with no marker).

The MDS-based schemes proposed and analysed in~\cite{ref:preuss2021efficient,ref:preuss2024sequencing} do not attempt to optimally process information available in $\by$. Instead, they make a hard decision on the symbol $\bx$ by choosing the $\numchosen$ most frequent motifs in $\by$. To mitigate interference, they require at least $t > 1$ copies of each motif in the top $\numchosen$ to be present in $\by$ and declare erasure if this condition is not met~\cite{ref:preuss2024sequencing}.
We estimate an upper bound on the capacity associated with this decoding approach by simulating the Coupon Collector Channel \textit{with no interference} and declaring erasure whenever $\by$ does not contain at least $t$ copies of each of the $\numchosen$ motifs in $\bx$. Having simulated the associated erasure probability,\footnote{An alternative way to compute this probability based on the Markov chain formulation of the coupon-collecting process is provided in~\cite{ref:preuss2024sequencing}.} which we denote by $\epsilon_t$, we calculate the capacity of the resulting non-binary erasure channel as
\begin{equation}
    \label{eq:capacity_nbec_t}
    \mathcal{C}_\NBEC(\nummotifs, \numchosen, t) = \log_2\binom{n}{k} \cdot (1 - \epsilon_t)\,.
\end{equation}
The resulting capacity estimations are shown in Fig.~\ref{fig:capacity_interference} (top) for $t = 2$ (purple curve with the square marker) and $t = 3$ (red curve with a triangle marker).
We observe that even without interference and for $t = 2$, the hard-decision approach analysed in~\cite{ref:preuss2024sequencing} exhibits a gap to $\CC(\nummotifs,\numchosen,\numreads,\pinter=0.078)$, which we shade green in Fig.~\ref{fig:capacity_interference} (top).
The coding scheme we propose in Section~\ref{sec:ldpc} operates in this region (black dot).
Moreover, the hard-decision schemes yield a symbol-level substitution---i.e., decide on a symbol $\hat{\bx} \neq \bx$---when \rfirst{a motif from $\mathcal{M} \setminus \bx$ is one of the top-$\numchosen$ most encountered motifs in $\by$ and is present} in at least $t$ copies. We simulate the rate of such symbol-level substitutions for $\CC(\nummotifs=8,\numchosen=4,\numreads,\pinter=0.078)$ and plot the corresponding curves in Fig.~\ref{fig:capacity_interference} (bottom) for $t = 2$ and $t = 3$ using the same line styles as in Fig.~\ref{fig:capacity_interference} (top).
These errors are not included in the analysis in~\cite{ref:preuss2024sequencing} because easily distinguishable motifs and the choice of $t > 1$ render them unlikely. While we agree with this in principle, we emphasise that $\pinter$ also depends on the probability of assigning a read to a wrong block, and we observe that the relatively high level of interference $\pinter = 0.078$ in the empirical data set results in non-negligible symbol error rates of about $10^{-2}$ for $t=2$ and $10^{-3}$ for $t=3$.

\vspace{-1ex}
\subsection{Read-Write Cost Trade-Off}
\label{sse:read_write_cost}

The read-write cost trade-off (first introduced explicitly in~\cite{ref:chandak2019improved}) is a fundamental property of DNA storage: the more redundancy is introduced to the system in the form of error-correcting codes (i.e., the lower the code rate), the fewer reads will suffice at the sequencing stage to reconstruct the data (i.e., the lower the required coverage).
Before we proceed to designing specific error-correcting schemes for the Coupon Collector Channels (Section~\ref{sec:ldpc}), we would like to highlight an important connection between the capacity curves and the read-write cost trade-off. We will use the notation and terminology of the Coupon Collector Channels and motif-based synthesis, although the same reasoning applies in equal measure to more standard nucleotide-by-nucleotide storage schemes.

The task of the system designer is to choose the code rate and the target coverage at the optimum point of the read-write \rfirst{cost} trade-off. The location of the optimum point depends on the costs of synthesising and sequencing a symbol. Denote by $\wcost~[\$/\text{cycle}]$ the cost of running one synthesis cycle and writing one combinatorial symbol, and by $\rcost~[\$/\text{read}]$ the cost of sequencing one motif belonging to that symbol. Suppose $C(\numreads)$ is a capacity curve that shows the highest information density (in information bits per cycle) achievable for a given \rfirst{number of reads} $\numreads$---implying that it would take $1/C(\numreads)$ cycles to store one information bit and $R / C(\numreads)$ reads to retrieve it.
The goal is to minimise the total cost of DNA storage per information bit,
\begin{equation}
    \label{eq:price}
    \begin{aligned}
        P(\numreads)  &= \wcost \cdot \frac{1}{C(\numreads)} + \rcost \cdot \frac{\numreads}{C(\numreads)} ~~~[\$/\text{inf. bit}] \\
        &= \rcost \cdot \frac{\rwcostratio + \numreads}{C(\numreads)} \,,
    \end{aligned}
\end{equation}
where $\rwcostratio = \wcost / \rcost$ describes how much more expensive synthesis is relative to sequencing. \rmine{This optimisation is introduced in~\cite{ref:heck2017fund}; in the following, we offer its visual interpretation based on the capacity curves.}
Assuming a well-behaved $C(\numreads)$,\footnote{A continuous version of $\mathcal{C}_\CC(\numreads)$ can be obtained by modelling $\numreads$ as a random variable and plotting  $\mathcal{C}_\CC$ against $\E{\numreads}$.} the optimum is obtained when
\begin{equation}
    \label{eq:price_derivative}
    P'(\numreads) = \rcost \cdot \frac{C(\numreads) - (\rwcostratio + \numreads)\,C'(\numreads)}{C(\numreads)^2} = 0\,,
\end{equation}
implying that at the optimum $\numreads$, denoted by $\numreadsopt$,
\begin{equation}
    \label{eq:rwcost_condition}
    C(\numreadsopt) = (\rwcostratio + \numreadsopt)\, C'(\numreadsopt)\,.
\end{equation}

\begin{figure}[!t]
    \centering
    \includegraphics[width=\columnwidth]{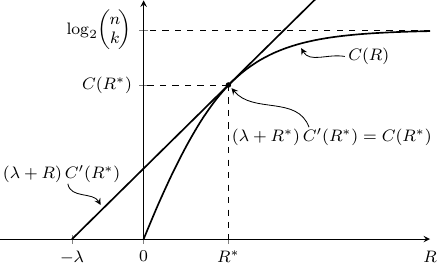}
    \vspace{-18pt}
    \caption{Capacity and the read-write cost trade-off. The optimum point $(\numreadsopt,C(\numreadsopt))$ is obtained according to~\eqref{eq:rwcost_condition}.}
    \label{fig:read_write}
    \vspace{-2ex}
\end{figure}

This optimisation is illustrated graphically in Fig.~\ref{fig:read_write}. For a candidate $\numreadsopt$, we can visualise a line $L(\numreads) = (\rwcostratio + \numreads)\,C'(\numreadsopt)$ that crosses the horizontal axis at $-\rwcostratio$ and has the same slope as $C(\numreads)$ at $\numreadsopt$. The larger $\numreadsopt$ is, the more the capacity curve $C(\numreads)$ flattens as it tends to the theoretical maximum $\log_2 \binom{\nummotifs}{\numchosen}$, and the line $L(\numreads)$ lowers down accordingly. The condition~\eqref{eq:rwcost_condition} means that the optimum $\numreadsopt$ corresponds to the case when the line just touches the capacity curve, and $\numreadsopt$ is equal to the point at which they meet.\footnote{The same procedure can be equivalently interpreted as finding a point at which the tangent line to $C(\numreads)$ crosses the horizontal axis at $-\rwcostratio$.}
Naturally, as the cost of writing relative to reading increases and $-\rwcostratio$ slides leftwards, $L(\numreads)$ touches $C(\numreads)$ at larger $\numreads$ and higher $C(\numreads)$.

Overall, we believe that this framework elucidates the relationship between capacity \rmine{curves} and the read-write cost trade-off: for example, it can be used to visually estimate what kind of read-write cost ratio is needed for a code rate to be relevant.

\section{Non-Binary LDPC Codes Over the Coupon Collector Channels}
\label{sec:ldpc}

To introduce error-correction coding into the motif-based DNA storage pipeline, we map motif combinations in $\cX$ to elements of the prime field $\mathrm{GF}(q)$, where $q$ is chosen to be the highest prime below $|\cX| = \binom{\nummotifs}{\numchosen}$. We use lexicographic ordering of motif combinations and leave the last $\binom{\nummotifs}{\numchosen} - q$ combinations unassigned.

We use non-binary LDPC \rfirst{(NB-LDPC) codes.} \rfirst{An NB-LDPC code} can be defined via a parity-check matrix $\mathbf{H} = (h_{ij})\rfirst{_{M \times N}}$ over $\mathrm{GF}(q)$ \rfirst{as the set of vectors $\bx_\mathsf{cw} = (\bx_1, \ldots, \bx_N)^{\mathsf{T}}$ that sa\-tis\-fy $\mathbf{H}\bx_\mathsf{cw} = 0$.}
In order to describe the decoding algorithms, we define the sets $N_i=\{j: h_{ij}\neq 0\}$ and $M_j=\{i: h_{ij}\neq 0\}$. The corresponding Tanner graph has an edge between variable node (VN) $v_j$ and check node (CN) $c_i$ if and only if $j\in N_i$ (or, equivalently, $i\in M_j$). In other words, $N_i$ and $M_j$ represent the local neighbourhoods of CN $c_i$ and VN $v_j$, respectively.
For the version of the channel without interference, we decode using a set-based belief propagation algorithm. For the version of the channel with interference, we use a probability-based belief propagation method dubbed QSPA (for $\boldsymbol{q}$-ary \textbf{s}um-\textbf{p}roduct \textbf{a}lgorithm)~\cite{ref:davey1998low},~\cite[Section 14.2]{ref:Ryan09}.

We remark that using only a subset of possible combinations has two important implications: First, it incurs a loss in the overall information transmission rate by a factor of $\log_2 q / \log_2\binom{\nummotifs}{\numchosen}$---for our running example of $(\nummotifs=8,\numchosen=4)$ and $\binom{\nummotifs}{\numchosen}=70$, the closest prime is $67$, and the (multiplicative) rate loss factor is $0.99$, implying a $1\%$ loss. Second, it breaks the symmetry of the channel---e.g., as we observe in our simulations, the all-zero codeword becomes more resilient to errors than other codewords. To circumvent this problem, we re-introduce symmetry in the channel by adding a pseudo-random mask modulo $|\cX| = \binom{\nummotifs}{\numchosen}$ to the codeword before transmission and undoing the effects of masking during decoding by the appropriate re-labelling of symbol possibilities and symbol likelihoods. In the following description of the decoding algorithms, we omit the trivial details of such re-labelling for clarity.

\vspace{-1ex}
\subsection{No Interference}
\label{sse:no_interference}


For the Coupon Collector Channel without interference (defined at the beginning of Section~\ref{sec:channel}), we employ a set-based decoder reminiscent of the message-passing Sudoku solver described in~\cite{ref:Atki14_sudoku}.
The decoder keeps track of the set of possibilities $S_j$ for each symbol and iteratively updates these sets based on parity-check constraints, progressively reducing the uncertainty about the transmitted symbols.
First, $S_j$ is initialised as the set of symbols in $\cX$ compatible with the corresponding channel output $\by_j$. If the number of unique motifs in $\by_j$ is $\ell$ (i.e., if $\by_j\in\cY_\ell$), then $|S_j| = \binom{\nummotifs - \ell}{\numchosen - \ell}$. The set of initial symbol possibilities $S_j$ for each symbol is passed to the corresponding VN $v_j$. We then iteratively apply the following procedure:

\begin{enumerate}
    \item \textit{CN update.} For convenience, we define a pairwise sum operation over finite sets of elements in $\mathrm{GF}(q)$ as \mbox{$A + B = \{a+b \, | \, a \in A, \, b \in B \}$} and a scalar multiplication operation as $c \cdot A = \{c \cdot a \, | \, a \in A\}$. For each \rfirst{VN} $v_j$, we generate the set of symbol possibilities $S_{ij}$ based on CN $c_i$ as
        \begin{equation}
            \label{eq:noint_cnupdate}
            S_{ij} = - h_{ij}^{-1} \sum_{t \in N_i \setminus \{j\}} h_{it} S_{t}\,.
        \end{equation}

    \item \textit{VN update.} Using $S_{ij}$ computed above, we update the set of symbol possibilities for each \rfirst{VN} $v_j$ as
        \begin{equation}
            \label{eq:noint_vnupdate}
            S_j \gets \bigcap_{i \in M_j} S_{ij} \cap S_j \,.
        \end{equation}
    \item \textit{Check for termination.} We repeat the decoding iterations \eqref{eq:noint_cnupdate}--\eqref{eq:noint_vnupdate} until either all VNs are resolved---i.e., the set of symbol possibilities for each VN is reduced to one element---or the total number of symbol possibilities remains the same for each VN, indicating decoding failure.
\end{enumerate}
\mycomment{
    Initial values, VN update, CN update. Set operations. Cite sudoku codes? Tree optimisation (perhaps a separate subsection)?
    }
    
\vspace{-1ex}
\subsection{Interference}
\label{sse:interference}


For the Coupon Collector Channel with interference (Section~\ref{sse:interference_capacity}), the belief propagation algorithm aims to approximate the symbol-wise posterior distribution of the transmitted codeword given the received information. This entails approximating a length-$q$ probability mass function for each codeword symbol. The output of the decoding algorithm is then given by the symbol-wise maximum \textit{a posteriori} estimate, provided the resulting codeword satisfies all parity constraints.

The QSPA works by performing message-passing over the Tanner graph corresponding to $\mathbf{H}$. Messages passed along the edges of the Tanner graph consist of $q$-ary probability mass functions. Let $Q_{ij} = (Q_{ij}^a)_{a \in \mathrm{GF}(q)}$ denote the message passed from $v_j$ to $c_i$, and $S_{ij} = (S_{ij}^a)_{a \in \mathrm{GF}(q)}$ denote the message in the other direction, from $c_i$ to $v_j$.

Decoding starts by initialising all $Q_{ij}$ vectors. For all $i,j$ such that $h_{ij}\neq 0$, we set $Q_{ij} = P_j$, where $P_j = (P_j^a)_{a\in\mathrm{GF}(q)}$ denotes the vector of conditional probabilities $\prob{\bx_j = a | \by_j}$ for $\bx_j$, the $j$-th symbol in the codeword $\bx_\mathsf{cw}$, and the corresponding channel output $\by_j$. The $P_j$ vectors are computed as described in Section~\ref{sse:interference_capacity} using~\eqref{eq:cond_prob_x_y}--\eqref{eq:cond_prob_y_x}. After initialisation, decoding proceeds by repeatedly iterating the following steps:

\begin{enumerate}
    \item \label{alg:qspa-step-1} \textit{CN update.} For each $a\in \mathrm{GF}(q)$, the update is
    \begin{equation}
        \label{eq:int_cnupdate}
        S_{ij}^a
        =
        \sum_{\bs}
        \mathds{1}\left(\sum_th_{it}\rfirst{s_t} + h_{ij}a = 0\right)
        \prod_t Q_{it}^{\rfirst{s_t}}\,,
    \end{equation}
    where the iteration variable \rfirst{$\bs = (s_t)$} runs over all $q$-ary sequences of length $|N_i| - 1$, indexed by $N_i \setminus \{j\}$ for notational convenience, and $t \in N_i \setminus \{j\}$.
    \item \textit{VN update.} For each $a\in \mathrm{GF}(q)$, compute first
    \begin{equation}
        \label{eq:int_vnupdate}
        Q_{ij}^a
        =
        P_j^a
        \prod_{t\in M_j \setminus \{i\}} S_{it}^a
    \end{equation}
    and then normalise each $Q_{ij}$ vector so that it sums to $1$.
    \item \textit{Check for termination.} \rfirst{Compute $\bz = (\bz_j)$ via}
    \begin{equation}
        \label{eq:int_termination}
        \bz_j = \arg\max_{a\in \mathrm{GF}(q)} P_j^a \prod_{t\in M_j} S_{it}^a\,.
    \end{equation}
    If $\mathbf{H}\bz=0$, return $\bz$ as the decoded codeword. Otherwise, go back to step \ref{alg:qspa-step-1}.
\end{enumerate}

This procedure may fail to output a codeword, so we set a maximum number of iterations, after which the algorithm declares decoding failure. 

It is important to point out that step \ref{alg:qspa-step-1} can be implemented efficiently via the Fast Fourier Transform (FFT) algorithm. Assume for simplicity that $\mathbf{H}$ consists only of $0$ and $1$ elements, as is the case in our experiments (Section~\ref{sec:results}), but the argument can be generalised to any $\mathbf{H}$ \rfirst{(see~\cite[Section 14.2]{ref:Ryan09})}. Then, the CN update rule in~\eqref{eq:int_cnupdate} can be rewritten as
\begin{equation}
\label{eq:update-s-messages-v2}
    S_{ij}^a
    =
    \sum_\bs \mathds{1}\left(\sum_t \rfirst{s_t} = -a \right)
    \prod_t Q_{it}^{\rfirst{s_t}}\,.
\end{equation}
This can be computed by first evaluating
\begin{equation}
    \label{eq:update-s-messages-conv}
    \Tilde{S}_{ij} = \bigotimes_{t\in N_i \setminus \{j\}} Q_{it}\,,
\end{equation}
where $\otimes$ denotes the circular convolution, and then re-indexing $\Tilde{S}_{ij}$ to obtain $S_{ij}^a = \Tilde{S}_{ij}^{-a}$ for all $a\in\mathrm{GF}(q)$. The FFT can be used to efficiently implement the series of circular convolutions in~\eqref{eq:update-s-messages-conv}.

\vspace{-1ex}
\subsection{Protograph-Based Non-Binary SC-LDPC Codes}
\label{sse:sc_ldpc}

Spatial coupling is a powerful technique that allows the construction of capacity-achieving codes under suboptimal but computationally efficient belief propagation (BP) decoding.
It involves interconnecting a sequence of Tanner graphs of the underlying uncoupled codes and creating a coupled chain where the termination boundaries contain CNs of lower average degrees.
The CNs at the boundaries help ignite the decoding process, and reliable information propagates from the boundaries of the chain inwards in a wave-like fashion.
Importantly, spatially coupled low-density parity-check (SC-LDPC) \mbox{codes~\cite{ref:Jime99,ref:Lent10}} have been proven to \textit{universally} achieve capacity over the broad class of binary-input memoryless symmetric channels~\cite{ref:Kude13}, which means that a sequence of codes that achieves capacity for one such channel is proven to achieve capacity for all of them.
Even though universality has not been formally proven for non-binary channels, it motivates us to choose non-binary SC-LDPC codes for Coupon Collector Channels and achieve excellent performance (Section~\ref{sec:results}) without channel-tailored code optimisation.

We use the classical protograph-based regular SC-LDPC code construction proposed in~\cite{ref:Srid07}. To obtain the protograph of a terminated $(\dv,\dc,\chainlength,\blocklength)$ SC-LDPC code, we start from a ``spatial'' sequence of $\chainlength$ copies of the protograph of the $(\dv,\dc)$-regular LDPC code and rearrange their edges such that a variable node (VN) at position $i \in \{ 1, \ldots, \chainlength \}$ is connected to $\dv$ consecutive positions $(i,i+1,\ldots,i+\dv-1)$, and a check node (CN) at position $i$ is connected to $\dc / \dv$ VNs at each of the $\dv$ positions $(i - \dv + 1, \ldots, i)$, where we assume $\dc / \dv$ is an integer. To connect the overhanging edges, $\dv - 1$ spatial positions that contain CNs only are appended to the sequence.
The Tanner graph is generated from the protograph using the standard copy-and-permute operation, whereby the SC-LDPC protograph is copied $\frac{\dv}{\dc}\blocklength$ times and the copies of the same edge in the original protograph are permuted.
The resulting Tanner graph contains \rfirst{$N = \chainlength\blocklength$} VNs and \rfirst{$M = \frac{\dv}{\dc}\blocklength(\chainlength + \dv - 1)$} CNs, yielding the design rate of
\begin{equation}
    \label{eq:rate_sc_ldpc}
    \rsc = 1 - \frac{\dv}{\dc}\cdot\frac{(\chainlength + \dv - 1)}{\chainlength}
    =  1 - \frac{\dv}{\dc}\left( 1 + \frac{\dv - 1}{\chainlength} \right) \,,
\end{equation}
which tends to the rate of the underlying LDPC ensemble, $1 - \dv/\dc\,$, as $\chainlength \to \infty$.
We refer the reader to~\cite{ref:Srid07,ref:Iyen12} for a more general description of protograph-based SC-LDPC codes.

\vspace{-1ex}
\subsection{Combinatorial Explosion in Decoding Complexity}
\label{sse:explosion}

One of the advantages of combinatorial motif-based DNA storage is that it enables a nearly linear growth in information density with the size of the motif library, in contrast with the more standard nucleotide-by-nucleotide synthesis where doubling the size of the alphabet only adds one bit to the maximum information density. Specifically, let us assume in this section that $\numchosen=\nummotifs/2$ so that $|\cX| = \binom{\nummotifs}{\numchosen}$ is maximised for a given $\nummotifs$. As $\nummotifs\to\infty$, the maximum information density grows as
\begin{equation}
    \label{eq:inf_density}
    \begin{aligned}
        \log_2|\cX| &= \log_2\binom{\nummotifs}{\nummotifs/2} \sim \logtwo{\frac{2^\nummotifs}{\sqrt{\nummotifs\pi/2}}} \\
        &= \nummotifs - \log_2\sqrt{\nummotifs\pi/2} = \nummotifs - \mathcal{O}(\log_2\nummotifs)\,.
    \end{aligned}
\end{equation}
Furthermore, increasing $\numchosen = \nummotifs / 2$ requires a higher $\numreads$ at the reading stage and makes the system more relevant for higher synthesis-sequencing cost ratios $\rwcostratio$, as we elaborated in Section~\ref{sse:read_write_cost}.

\begin{figure}[!t]
    \centering
    \includegraphics[width=\columnwidth]{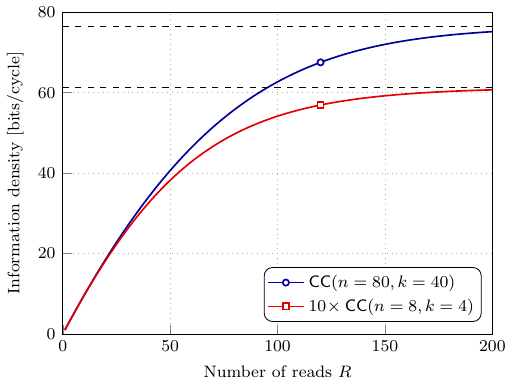}
    \vspace{-18pt}
    \caption{Capacity of the $\CC(\nummotifs=80,\numchosen=40,\numreads)$ channel (blue curve, circle marker) and of the library-splitting scheme with $a = 10$ (red curve, square marker). The gap between the curves illustrates the price paid for manageable computational complexity.}
    \label{fig:capacity_splits}
    \vspace{-2ex}
\end{figure}

However, increasing $\numchosen$ and $\nummotifs$ presents a significant challenge to the decoding methods proposed in Sections~\ref{sse:no_interference} and~\ref{sse:interference} above: their decoding complexity depends on $|\cX| = \binom{\nummotifs}{\numchosen}$ that grows exponentially with increasing $(\nummotifs,\numchosen\!=\!\nummotifs/2)$ as shown in~\eqref{eq:inf_density}---in other words, the very mechanism that makes combinatorial motif-based storage attractive renders straightforward application of the proposed coding methods impractical.

We propose the following approach to circumvent this problem: Suppose the largest alphabet the decoder can handle is a choice of $\tilde{\numchosen}$ out of $\tilde{\nummotifs}$ motifs. We split the target library of $\nummotifs = a \tilde{\nummotifs}$ motifs into $a$ groups of $\tilde{\nummotifs}$ and restrict the choice of $\numchosen = a \tilde{\numchosen}$ motifs to $\tilde{\numchosen}$ motifs \textit{within each group} instead of a choice of any $\numchosen$ motifs out of $\nummotifs$. Such restriction lowers the number of possible motifs and hence information density, but on the other hand it allows to operate the decoder as if each cycle contains $a$ independent symbols from a smaller alphabet. Crucially, since the total number of possible symbols in a cycle becomes
\begin{equation}
    \label{eq:explosion_alphabet_size}
    |\cX| = \binom{\tilde{\nummotifs}}{\tilde{\numchosen}}^{\nummotifs/\tilde{\nummotifs}}\,,
\end{equation}
the linear growth in maximum information density is retained:
\begin{equation}
    \label{eq:linear_growth_explosion}
    \log_2|\cX| = \log_2\binom{\tilde{\nummotifs}}{\tilde{\numchosen}}^{\nummotifs/\tilde{\nummotifs}}
    = \nummotifs \cdot \frac{\log_2\binom{\tilde{\nummotifs}}{\tilde{\numchosen}}}{\tilde{\nummotifs}} \,,
\end{equation}
albeit with a slower growth rate of $\log_2\binom{\tilde{\nummotifs}}{\tilde{\numchosen}}/\tilde{\nummotifs}$ instead of $1$ as in~\eqref{eq:inf_density}.
For example, for $\tilde{\nummotifs}=8$ and $\tilde{\numchosen}=4$ as in our examples earlier, the penalty term in~\eqref{eq:linear_growth_explosion} evaluates to about $0.77$, which means that asymptotically the scheme loses around $23\%$ in maximum information density relative to the computationally infeasible processing of all $\binom{\nummotifs}{\numchosen}$ possible combinations.

The capacity of the library-splitting scheme is obtained by averaging the capacity of an individual sub-channel over the binomial distribution of the effective number of reads. Specifically, taking the interference-free Coupon Collector Channel as an example, we get
\begin{equation}
    \label{eq:capacity_splits}
    \mathcal{C}_{a \times \CC} =
    a \sum_{r=1}^{R} \mathcal{C}_\CC(\tilde{\nummotifs},\tilde{\numchosen},r) \binom{\numreads}{i} \left( \frac{1}{a} \right)^r \left( 1 - \frac{1}{a} \right)^{\numreads - r}
\end{equation}
with $\mathcal{C}_\CC(\tilde{\nummotifs},\tilde{\numchosen},r)$ denoting the capacity of $\CC(\nummotifs = \tilde{\nummotifs},\numchosen = \tilde{\numchosen},R = r)$ evaluated via~\eqref{eq:capacity_cc}.

Fig.~\ref{fig:capacity_splits} compares the capacity curve for the $\CC(\nummotifs=80,\numchosen=40,\numreads)$ channel (blue curve with the circle marker) obtained using~\eqref{eq:capacity_cc} with the capacity curve for the case where the library of $\nummotifs=80$ motifs is split into $a = 10$ groups of $\tilde{\nummotifs}=8$ motifs each with $\tilde{\numchosen}=4$ (red curve with the square marker) obtained using~\eqref{eq:capacity_splits}. We observe that maintaining a manageable decoding complexity using the proposed scheme induces a penalty on information density. The penalty grows for increasing $\numreads$ and eventually stabilises at the asymptotic value dictated by~\eqref{eq:linear_growth_explosion}---in this case at around a quarter of the maximum information density. It is an interesting open problem to reduce this gap by, e.g., constructing a coding scheme that allows decoding without keeping track of all $|\cX| = \binom{\nummotifs}{\numchosen}$ possible symbols and thus avoids combinatorial explosion in decoding complexity altogether.

\section{Numerical Results}
\label{sec:results}

\begin{figure}[!t]
    \centering
    \includegraphics[width=\columnwidth]{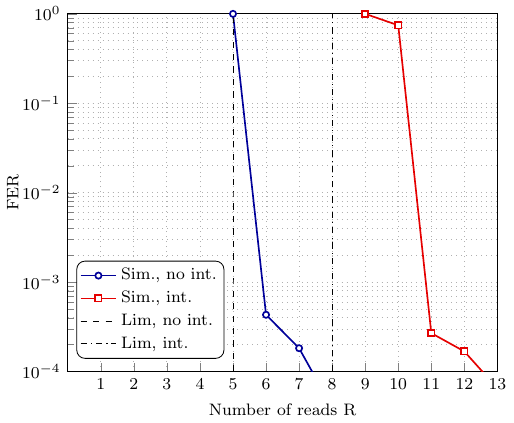}
    \vspace{-18pt}
    \caption{FER for the $(\dv\!=\!4, \dc\!=\!12, \chainlength\!=\!50, \blocklength\!=\!1002)$ SC-LDPC code ensemble with $q=67$ over the Coupon Collector Channel with $(\nummotifs=8,\numchosen=4)$, with no interference (blue, circle markers) and with $\pinter=0.078$ (red, square markers). For the employed overall information transmission rate of $3.92$ bits/cycle, the minimum achievable $\numreads$ (based on the capacity curves in Figs.~\ref{fig:capacity} and~\ref{fig:capacity_interference}) for the channel with and without interference is shown as the dash-dotted and dashed vertical line, respectively.}
    \label{fig:cc_result}
    \vspace{-2ex}
\end{figure}

To demonstrate the potential of the proposed coding scheme, we use the protograph-based regular $(\dv=4, \dc=12, \chainlength=50, \blocklength=1002)$ SC-LDPC code ensemble (see Section~\ref{sse:sc_ldpc}). For simplicity, we only consider binary entries in the parity-check matrix $\mathbf{H}$. We use $q=67$, targeting $(\nummotifs=8,\numchosen=4)$, which results in an overall information rate of $\rsc \cdot \log_2 q \approx 3.92$, with $\rsc$ calculated in~\eqref{eq:rate_sc_ldpc}.

Fig.~\ref{fig:cc_result} shows the simulated frame error rate (FER) of this code ensemble over $\CC(\nummotifs=8,\numchosen=4,\numreads, \pinter=0)$ (blue line with circle markers) and $\CC(\nummotifs=8,\numchosen=4,\numreads, \pinter=0.078)$ (red line with square markers). The minimum achievable $\numreads$ for these two channels and the employed code rate of $3.92$ is shown as the dashed and dash-dotted line, respectively.\footnote{
The minimum achievable $\numreads$ is the smallest $\numreads$ for which $C(\numreads) > 3.92$.}

We observe that the ensemble operates at a FER below $10^{-3}$ for $\numreads=6$ over the channel without interference and for $\numreads=11$ over the channel with $\pinter=0.078$. \rmine{The corresponding points are shown} in Figs.~\ref{fig:capacity} and~\ref{fig:capacity_interference} (black dots), well within the green-shaded regions. Notably, the capacity-approaching performance (especially for the interference-free case) was achieved without ensemble optimisation, which demonstrates the power of spatial coupling in conjunction with channel-tailored message-passing decoding.

\section{Conclusion}
\label{sec:conclusion}

We explore a practically interesting DNA storage scheme based on combinatorial motif-based synthesis. We analyse a data set from a real-world experiment and devise a model of the channel that mimics the empirically observed behaviour. At the core of the channel model is the Coupon Collector’s Problem, which was highlighted in a similar setting earlier by Preuss \textit{et al.} in~\cite{ref:preuss2021efficient,ref:preuss2024sequencing}. In contrast to their work, however, here we propose to process partial information available out of the channel instead of limiting decoding only to fully recovered symbols, which allows us to approach the fundamental capacity limits that are outside of the reach of the earlier schemes. We show how to calculate those fundamental limits and how to reason about the associated trade-off between the read and write costs. We devise an error-correcting scheme that approaches these limits in practice and show how to circumvent the combinatorial complexity explosion in our decoding algorithm at the cost of some additional redundancy. Overall, the paper offers a powerful framework to reason about and incorporate error correcting coding in one of the most promising approaches to DNA data storage. 

\balance

\bibliographystyle{IEEEtran}

\end{document}